\documentclass[11pt, a4paper]{article}

\pdfoutput=1

\usepackage[utf8]{inputenc}
\usepackage[T1]{fontenc}
\usepackage[colorlinks]{hyperref}
\usepackage{microtype}

\hypersetup{
  linkcolor=[rgb]{0.3,0.3,0.6},
  citecolor=[rgb]{0.2, 0.6, 0.2},
  urlcolor=[rgb]{0.6, 0.2, 0.2}
}

\usepackage{mathtools, amsthm, amsfonts, amssymb}
\usepackage{booktabs}
\usepackage{braket}
\usepackage{commath}
\usepackage{caption}
\usepackage{titlesec}
\titlelabel{\thetitle. }
\titleformat*{\section}{\large\bfseries}

\usepackage[capitalise]{cleveref}

\newtheorem{theorem}{Theorem}
\newtheorem{proposition}[theorem]{Proposition}
\newtheorem{lemma}[theorem]{Lemma}

\theoremstyle{definition}
\newtheorem{definition}[theorem]{Definition}

\newtheorem*{problem*}{Problem}

\newtheorem{remark}[theorem]{Remark}

\DeclareMathOperator{\supprank}{R_s}
\DeclareMathOperator{\bsupprank}{\underline{R}_s}
\DeclareMathOperator{\rank}{R}
\DeclareMathOperator{\borderrank}{\underline{R}}

\newcommand{\CC}{\mathbb{C}}
\newcommand{\NN}{\mathbb{N}}
\newcommand{\Sym}{\mathrm{Sym}}
\newcommand{\GL}{\mathrm{GL}}

\newcommand{\GHZ}{\mathrm{GHZ}}

\newcommand{\id}{\mathbf{1}}

\newcommand{\ZZ}{\mathbb{Z}}

\DeclareMathOperator{\Span}{Span}

\DeclareMathOperator{\diag}{diag}

\newcommand{\FF}{\mathbb{F}}

\DeclareMathOperator{\dett}{det}

\newcommand{\Oh}{\mathcal{O}}

\newcommand{\pert}{q}

\begin{document}

\vspace*{0.7em}
\begin{center}
\Large\textbf{The border support rank of two-by-two\\[0.2em] matrix multiplication is seven}\par
\vspace{1em}
\large Markus Bläser, Matthias Christandl and Jeroen Zuiddam\par
\end{center}
\vspace{0.7em}

\begin{abstract}
We show that the border support rank of the tensor corresponding to two-by-two matrix multiplication is seven over the complex numbers. We do this by constructing two polynomials that vanish on all complex tensors with format four-by-four-by-four and border rank at most six, but that do not vanish simultaneously on any tensor with the same support as the two-by-two matrix multiplication tensor. This extends the work of Hauenstein, Ikenmeyer, and Landsberg. We also give two proofs that the support rank of the two-by-two matrix multiplication tensor is seven over any field: one proof using a result of De Groote saying that the decomposition of this tensor is unique up to sandwiching, and another proof via the substitution method. These results answer a question asked by Cohn and Umans. Studying the border support rank of the matrix multiplication tensor is relevant for the design of matrix multiplication algorithms, because upper bounds on the border support rank of the matrix multiplication tensor lead to upper bounds on the computational complexity of matrix multiplication, via a construction of Cohn and Umans. Moreover, support rank has applications in quantum communication complexity.
%Second, support rank characterizes the communication complexity of a function in the multiparty nondeterministic quantum broadcast model of Buhrman et~al. In particular, our results imply that Strassen's algorithm yields the optimal nondeterministic quantum communication protocol for the three-player cyclic equality game with one-bit inputs.
\end{abstract}

%\paragraph{Keywords.} Tensor rank, border rank, algebraic complexity theory, matrix multiplication.

\section{Introduction}

%In linear algebra, the matrix product of the $n\times n$ matrix $A$ with entries $A_{ij}$ and the $n\times n$ matrix $B$ with entries $B_{ij}$ is defined as the matrix $AB$ with entries $(AB)_{ik} = \sum_{j=1}^n a_{ij}b_{jk}$. Writing $\FF^{n\times n}$ for the $\FF$-vector space of $n\times n$ matrices, the matrix multiplication map is the map $\FF^{n\times n} \times \FF^{n\times n} \to \FF^{n\times n}$ that maps $(A,B) \mapsto AB$. Trivially, this  Since the work of Strassen we know that there exist algorithms that compute this map in 
%
%The matrix multiplication map is the map that takes two $n\times n$ matrices $A$ and $B$, with entries in some field $\FF$, and outputs the matrix product $AB$ whose entries are defined by $(AB)_{ik} = \sum_{j=1}^n a_{ij}b_{jk}$. Writing $\FF^{n\times n}$ for the $\FF$-vector space of $n\times n$ matrices with entries in $\FF$, the matrix multiplication map is an $\FF$-bilinear map $\FF^{n\times n} \times \FF^{n\times n} \to \FF^{n\times n}$.

Multiplication of two $n\times n$ matrices over a field $\FF$ is an $\FF$-bilinear map ${\FF^{n\times n} \times \FF^{n\times n} \to \FF^{n\times n}}$ called the matrix multiplication map. The matrix multiplication map corresponds naturally to the following structure tensor. Let $[n]$ be the set $\{1,2,\ldots, n\}$ and let $\{e_{ij} : i,j\in [n]\}$ be the standard basis for the vector space $\FF^{n\times n}$ of $n\times n$ matrices. Define the structure tensor of the matrix multiplication map as
\[
\langle n,n,n \rangle \coloneqq \sum_{\mathclap{i,j,k\in [n]}} e_{ij}\otimes e_{jk} \otimes e_{ki}\, \in\, \FF^{n\times n} \otimes \FF^{n\times n} \otimes \FF^{n\times n}.
\]
%be the corresponding structure tensor, where $[n]$ is the set $\{1,2,\ldots, n\}$.
(Technically, this is the structure tensor of the trilinear map that computes the trace of a product of three matrices.)
Let $V_1$, $V_2$, and $V_3$ be vector spaces. The tensor rank of a tensor $t\in V_1\otimes V_2\otimes V_3$ is the smallest number $r$ such that~$t$ can be written as a sum of $r$ simple tensors $v_1\otimes v_2\otimes v_3 \in V_1 \otimes V_2 \otimes V_3$.
The computational complexity of matrix multiplication is tightly related to the tensor rank of the tensor $\langle n,n,n\rangle$ (see e.g.~\cite{burgisser1997algebraic}). Strassen showed that the tensor rank of $\langle 2,2,2 \rangle$ is at most seven over any field \cite{strassen1969gaussian}; Hopcroft and Kerr \cite{hopcroft1971minimizing} showed that the tensor rank is at least seven over the finite field $\FF_2$, and Winograd \cite{winograd1971multiplication} showed that the tensor rank is at least seven over any field. Over an algebraically closed field, the border rank of a tensor $t\in V_1\otimes V_2\otimes V_3$ is the smallest number $r$ such that~$t$ is in the Zariski closure of all tensors of rank at most $r$ in $V_1 \otimes V_2 \otimes V_3$.  Landsberg proved that the border rank of $\langle 2,2,2 \rangle$ is seven over the field $\CC$ of complex numbers \cite{landsberg2006border}, and a different proof for this based on highest-weight vectors was later given by Hauenstein, Ikenmeyer and Landsberg \cite{hauenstein2013equations}.

We extend the above results. Let $t\in V_1\otimes V_2 \otimes V_3$ be a tensor in a fixed basis, a hypermatrix. The \emph{support} of $t$ is the set of coordinates where $t$ has a nonzero coefficient. The \emph{support rank} of $t$ is the minimal rank of a tensor with the same support as $t$. This has also been called s-rank \cite{cohn2013fast}, nondeterministic rank \cite{de2003nondeterministic}, zero-one rank \cite{wigdersontalk} and minimum rank of a nonzero pattern \cite{berman2008minimum} in the literature. The \emph{border support rank} of $t$ is the minimal border rank of a tensor with the same support as $t$. We prove the following.

\begin{theorem}\label{subth}
The support rank of $\langle 2,2,2 \rangle$ is seven over any field $\FF$.
\end{theorem}

\begin{theorem}\label{mainth}
The border support rank of $\langle 2,2,2 \rangle$ is seven over $\CC$.
\end{theorem}

\cref{subth} and \cref{mainth} answer a question of Cohn and~Umans~\cite{cohn2013fast}, that was also posed as an open problem during the Algorithms and Complexity in Algebraic Geometry programme at the Simons Institute \cite{simons}. We note that, in general, computing the tensor rank or support rank of a tensor is a computationally hard task. Namely, given a 3-tensor $t$ and a natural number~$r$, deciding whether the tensor rank of $t$ is at most $r$ is NP-complete over any finite field \cite{HASTAD1990644} and NP-hard over any integral domain \cite{shitov2016hard}.
Moreover, given a 2-tensor (that is, a matrix) $A$ and a natural number~$r$, deciding whether the support rank of $A$ is at most $r$ is NP-hard over the real numbers \cite{bhangale2015complexity}.

Previously, it was known that the border support rank of the matrix multiplication tensor $\langle n,n,n \rangle$ is at least $2n^2-n$ \cite{buhrman2016nondeterministic}, so in particular that the border support rank of $\langle 2,2,2 \rangle$ is at least six. This result was obtained using Young flattenings.

% Move up
Studying (border) support rank is interesting for two reasons. The first reason comes from algebraic complexity theory. As mentioned above, the tensor rank of the matrix multiplication tensor is tightly related to the computational complexity of matrix multiplication. It turns out that asymptotically, the border support rank of matrix multiplication gives an upper bound on the tensor rank of matrix multiplication, as follows.
%
%Cohn and Umans study the support rank of the matrix multiplication tensor in the context of the exponent of matrix multiplication.
% 
The exponent of matrix multiplication~$\omega$ is defined as the smallest number $\beta$ such that for any $\varepsilon>0$ the tensor rank of $\langle n,n,n \rangle$ is in $\Oh(n^{\beta+\varepsilon})$. The number $\omega$ is between 2 and 2.3728639 \cite{le2014powers} and it is a major open problem in algebraic complexity theory to decide whether $\omega$ equals 2. One can define an analogous quantity $\omega_s$ for the support rank of $\langle n,n,n \rangle$. One can show with Strassen's laser method that $\omega \leq (3\omega_s - 2)/2$ \cite{cohn2013fast}. To show that $\omega = 2$, it therefore suffices to show that $\omega_s = 2$. Cohn and Umans aim to obtain upper bounds on $\omega_s$ by realizing the algebra of $n\times n$ matrices inside some cleverly chosen group algebra.

The second reason, which was our original motivation, comes from quantum communication complexity. Let $f: X\times Y \times Z \to \{0,1\}$ be a function on a product of finite sets $X$, $Y$ and $Z$. Alice, Bob and Charlie have to compute $f$ in the following sense. Alice receives an $x\in X$, Bob receives a $y\in Y$ and Charlie receives a $z\in Z$. Moreover, the players share a so-called Greenberger-Horne-Zeilinger (GHZ) state of rank $r$, which is described by the tensor $\GHZ_r = \sum_{i=1}^r e_i \otimes e_i \otimes e_i \in (\CC^r)^{\otimes 3}$.
The players apply local quantum operations. After this, each player has to output a bit such that if $f(x,y,z) = 1$, then with some nonzero probability all players output 1 and if $f(x,y,z) = 0$, then with probability zero all players output 1. The complexity of such a protocol is the logarithm of the rank $r$ of the GHZ-state used, and the minimum complexity of all quantum protocols for $f$ is the \emph{nondeterministic communication complexity} of~$f$. This number equals the logarithm of the support rank of the tensor with support given by~$f$, that is $\sum_{x,y,z} f(x,y,z)\, e_x\otimes e_y \otimes e_z$ \cite{buhrman2016nondeterministic}. Similarly, the logarithm of the border support rank of the tensor with support given by $f$ equals the \emph{approximate} nondeterministic communication complexity of $f$.
%
%Key to the communication complexity result is the fact that tensors can be viewed as quantum states and that the tensor rank can be viewed as the minimal size (in terms of levels $r$) of a GHZ state required to generate this quantum state by stochastic local operations and classical communication (SLOCC). The matrix multiplication tensor turns out to correspond to pairwise Einstein-Podolsky-Rosen pairs distributed between three players \cite{chitambar2008tripartite}. 
Since tensor rank and border rank are natural measures of entanglement, our result may also be of interest to the quantum information theory community. 

\paragraph{Notation.} For any tensor $t$, we will denote tensor rank by $\rank(t)$, border rank by $\borderrank(t)$, support rank by $\supprank(t)$ and border support rank by $\bsupprank(t)$.

%We note that one can show the following general bounds using Young flattenings by making small modifications to proofs of Landsberg and Ottaviani.

%\begin{theorem}\label{genborderth}
%The border support rank of $\langle n,n,n \rangle$ is at least $2n^2 - n$. TODO: we still need to prove something here.
%\end{theorem}

%\begin{theorem}\label{genth}
%The support rank of $\langle n,n,n \rangle$ is at least ?. TODO: for this we should dive into the proof of Landsberg.
%\end{theorem}

%The lower bound in \cref{genborderth} improves the lower bound $\tfrac32 n^2 - 2$ for $n$ even and $\tfrac32n^2 - \tfrac12$ for $n$ odd of Ikenmeyer \cite[Proposition 8.2.16]{ikenmeyer2013geometric}.

\paragraph{Paper outline.}
This paper is structured as follows. In \cref{sec:supprank} we give two proofs for \cref{subth}. In \cref{sec:bsupprank} we give a short introduction to border rank lower bounds by highest-weight vectors and then apply this theory to prove \cref{mainth}.

\section{Support rank}\label{sec:supprank}

%We denote the tensor rank of $t$ by $\rank(t)$ and the support rank by $\supprank(t)$.
We will give two proofs for \cref{subth}. Both proofs use the following lemma that reduces the 8-parameter minimization problem at hand to a 1-parameter minimization problem. Let~$\FF$ be a field.  Let $e_{11} = \bigl(\begin{smallmatrix}1&0\\0&0\end{smallmatrix}\bigr)$, $e_{12}=\bigl(\begin{smallmatrix}0&1\\0&0\end{smallmatrix}\bigr)$, $e_{21}=\bigl(\begin{smallmatrix}0&0\\1&0\end{smallmatrix}\bigr)$, $e_{22}=\bigl(\begin{smallmatrix}0&0\\0&1\end{smallmatrix}\bigr)$ be the standard basis of the space of $2\times 2$ matrices~$\FF^{2\times 2}$ over $\FF$. Let $e_{1},e_{2},e_{3},e_{4}$ be the standard basis of $\FF^4$. We naturally identify $\FF^{2\times 2}$ with $\FF^4$ by $e_{11} \mapsto e_1$, $e_{12} \mapsto e_2$, $e_{21} \mapsto e_3$, $e_{22} \mapsto e_4$. Let $\GL_4(\FF)^{\times 3}$ act on the tensor space $\FF^{2\times 2}\otimes \FF^{2\times 2}\otimes\FF^{2\times 2}$ accordingly.
%Besides the natural group action of $\GL_4(\FF)^{\times 3}$ on $(\FF^{2\times2})^* \otimes (\FF^{2\times 2})^* \otimes \FF^{2\times 2}$ we will use the \emph{sandwiching action} of $\GL_2(\FF)^{\times 3}$ defined by
%\[
%(A,B,C) \cdot (X^* \otimes Y^* \otimes Z) = (AXB^{-1}, B
%\]

\begin{lemma}[Parameter reduction]\label{paramreduction}
Let $t\in \FF^{2\times2} \otimes \FF^{2\times2} \otimes \FF^{2\times2}$ be a tensor with the same support as the matrix multiplication tensor $\langle 2,2,2\rangle$. There is a tensor $s$ in the $\GL_4(\FF)^{\times 3}$-orbit of $t$, with the same support as~$t$, such that all nonzero entries of $s$ are 1 except possibly for the coefficient of~$e_{11}\otimes e_{11} \otimes e_{11}$.
\end{lemma}
\begin{proof}
Identify the tensor $\langle 2,2,2\rangle = \sum_{i,j,k\in [2]} e_{ij}\otimes e_{jk} \otimes e_{k\ell}$ with the tensor
\[
e_{111} + e_{123} + e_{231} + e_{243} + e_{312} + e_{324} + e_{432} + e_{444} \in \FF^4 \otimes \FF^4 \otimes \FF^4,
\]
where $e_{ijk} = e_i \otimes e_j \otimes e_k$. We can view this tensor as as a $4\times 4\times 4$ cube filled with elements $0$ and $1$ from $\FF$.
Let $t$ be a tensor in $\FF^4 \otimes \FF^4 \otimes \FF^4$ with  the same support as $\langle 2,2,2 \rangle$, so, in 1-slices,
\begin{align*}
t &= \quad
\begin{bmatrix} 
a & 0 & 0 & 0\\
0 & 0 & b & 0\\
0 & 0 & 0 & 0\\
0 & 0 & 0 & 0
\end{bmatrix}
\quad
\begin{bmatrix} 
0 & 0 & 0 & 0\\
0 & 0 & 0 & 0\\
c & 0 & 0 & 0\\
0 & 0 & d & 0
\end{bmatrix}
\quad
\begin{bmatrix} 
0 & e & 0 & 0\\
0 & 0 & 0 & f\\
0 & 0 & 0 & 0\\
0 & 0 & 0 & 0
\end{bmatrix}
\quad
\begin{bmatrix} 
0 & 0 & 0 & 0\\
0 & 0 & 0 & 0\\
0 & g & 0 & 0\\
0 & 0 & 0 & h
\end{bmatrix}\\
\intertext{where $a,b,c,d,e,f,g,h$ are nonzero elements in $\FF$. Here we index the 1-slices by the first tensor leg, the rows of the slices by the second tensor leg and columns of the slices by the third tensor leg. Scaling the 1-slices of $t$ according to $\diag(1/b, 1/d, 1/f, 1/h)$, that is, applying $\diag(1/b, 1/d, 1/f, 1/h) \otimes \id_4\otimes \id_4$ to $t$, yields a tensor of the form}
t' &= \quad
\begin{bmatrix} 
a' & 0 & 0 & 0\\
0 & 0 & 1 & 0\\
0 & 0 & 0 & 0\\
0 & 0 & 0 & 0
\end{bmatrix}
\quad
\begin{bmatrix} 
0 & 0 & 0 & 0\\
0 & 0 & 0 & 0\\
c'& 0 & 0 & 0\\
0 & 0 & 1 & 0
\end{bmatrix}
\quad
\begin{bmatrix} 
0 & e'& 0 & 0\\
0 & 0 & 0 & 1\\
0 & 0 & 0 & 0\\
0 & 0 & 0 & 0
\end{bmatrix}
\quad
\begin{bmatrix} 
0 & 0 & 0 & 0\\
0 & 0 & 0 & 0\\
0 & g' & 0 & 0\\
0 & 0 & 0 & 1
\end{bmatrix}\\
\intertext{Scaling the rows of $t'$ according to $\diag(1/e', 1, 1/g', 1)$, that is, applying $\id_4 \otimes \diag(1/e', 1, 1/g', 1) \otimes \id_4$ to $t'$, yields a tensor of the form}
t'' &= \quad
\begin{bmatrix} 
a'' & 0 & 0 & 0\\
0 & 0 & 1 & 0\\
0 & 0 & 0 & 0\\
0 & 0 & 0 & 0
\end{bmatrix}
\quad
\begin{bmatrix} 
0 & 0 & 0 & 0\\
0 & 0 & 0 & 0\\
c''& 0 & 0 & 0\\
0 & 0 & 1 & 0
\end{bmatrix}
\quad
\begin{bmatrix} 
0 & 1 & 0 & 0\\
0 & 0 & 0 & 1\\
0 & 0 & 0 & 0\\
0 & 0 & 0 & 0
\end{bmatrix}
\quad
\begin{bmatrix} 
0 & 0 & 0 & 0\\
0 & 0 & 0 & 0\\
0 & 1 & 0 & 0\\
0 & 0 & 0 & 1
\end{bmatrix}
\end{align*}
Finally, scaling the columns of $t''$ according to $\diag(1/c'',1,1,1)$, that is, applying $\id_4 \otimes \id_4 \otimes \diag(1/c'',1,1,1)$ to $t''$, yields a tensor of the required form.
\end{proof}

Our first proof of \cref{subth} uses a corollary of a result of De Groote on the uniqueness of the decomposition of $\langle 2,2,2\rangle$ into simple tensors. 
\begin{theorem}[{\cite[Remark 4.2]{de1978varieties}}]\label{sandwich}
Let $v\coloneqq v_1 \otimes v_2 \otimes v_3\in \FF^{2\times 2} \otimes \FF^{2\times 2} \otimes \FF^{2\times 2}$ be an element of an arbitrary optimal decomposition of $\langle 2,2,2 \rangle$ into simple tensors over $\FF$ such that the rank of each $v_i$ as an element of $\FF^2 \otimes \FF^2$ is one. Then there exist invertible matrices $A,B,C\in \GL_2(\FF)$ such that $v = (A^{-1}e_{11}B) \otimes (B^{-1}e_{11}C) \otimes (C^{-1}e_{11}A)$, where  $A$, $B$ and $C$ act by matrix multiplication from the left and right on $\FF^{2\times 2}$.
\end{theorem}

\begin{definition}\label{perturbed}
For any number $q \in \FF$, define the perturbed matrix multiplication tensor $\langle 2,2,2 \rangle_\pert \coloneqq \langle 2,2,2\rangle + (\pert-1)\, e_{11}\otimes e_{11}\otimes e_{11}$; this is the tensor obtained from $\langle 2,2,2\rangle$ by replacing the coefficient of $e_{11}\otimes e_{11} \otimes e_{11}$ by~$\pert$. 
\end{definition}
We now give our first proof of \cref{subth} using the above uniqueness statement.

\begin{proof}[\bfseries\upshape Proof of \cref{subth}; uniqueness argument] As was already observed by De Groote, \cref{sandwich} 
gives the upper bound $\rank(\langle 2,2,2 \rangle_0) \leq 6$ and thus $\rank(\langle 2,2,2\rangle_\pert) \leq 7$ for all $\pert\in \FF$.  %\cref{mainlem} gives the lower bound 
We claim that $\rank(\langle 2,2,2\rangle_\pert) \geq 7$ for all nonzero $\pert\in \FF$.
Suppose $\pert$ is a number in $\FF$ such that $\rank(\langle 2,2,2\rangle_\pert) = 6$. Let $\langle 2,2,2\rangle_\pert = \sum_{i=1}^6 u_i \otimes v_i \otimes w_i$ be a decomposition into simple tensors. Then
\[
\langle 2,2,2 \rangle = \langle 2,2,2\rangle_\pert + (1-\pert)\, e_{11}\otimes e_{11} \otimes e_{11} = \sum_{i=1}^6 u_i \otimes  v_i \otimes w_i + (1-\pert)\, e_{11}\otimes e_{11} \otimes e_{11}
\]
is an optimal decomposition of $\langle 2,2,2\rangle$ into simple tensors. Therefore, by \cref{sandwich}, %the term $(1-z) \kets{11}{11}{11}$ lies in the sandwiching-orbit of $\kets{11}{11}{11}$, which means that 
there exist $A,B,C$ in~$\GL_2(\FF)$ such that
\[
A^{-1} e_{11} B  \otimes  B^{-1} e_{11} C  \otimes  C^{-1} e_{11} A = (1-\pert)\, e_{11}\otimes e_{11} \otimes e_{11}.
\] %However, such $A,B,C$ do not exist. Indeed, suppose $A^{-1} \ket{11} B  \otimes  B^{-1} \ket{11} C  \otimes  A \ket{11} C^{-1} = \lambda \kets{11}{11}{11}$. 
Let $f_1, f_2$ be the standard basis of $\FF^2$. Then by taking appropriate transposes the previous equation is equivalent to
\[
A^{-1} f_1 \otimes B^T f_1  \otimes  B^{-1} f_1 \otimes C^T f_1  \otimes C^{-1} f_1 \otimes  A^T f_1  = (1-\pert)\, f_1^{\otimes 6}, 
\]
which implies that $A^T, B^T, C^T$ each have eigenvector $f_1$. Let $\alpha, \beta, \gamma$ be the respective eigenvalues. Then $A^{-1}, B^{-1}, C^{-1}$ have eigenvalues $\alpha^{-1}, \beta^{-1}, \gamma^{-1}$. This yields the equation $\alpha^{-1}\,\beta\,\beta^{-1}\,\gamma\,\alpha\,\gamma^{-1} = 1-\pert$.
We conclude that $\pert = 0$.
By \cref{paramreduction} we can conclude that $\supprank(\langle 2,2,2 \rangle) = 7$.
\end{proof}

Our second proof of \cref{subth} uses a method called the substitution method. Let $x_{ij}$, $y_{ij}$, $z_{ij}$ ($i,j\in [2]$) be variables. Let $X,Y,Z$ be the corresponding $2\times 2$ variable matrices. For $\pert \in \FF$, define the function
\[
f_\pert(X,Y,Z) \coloneqq \sum_{i,j,k\in[2]} x_{ij}\,y_{jk}\,z_{ki} + (\pert-1)\, x_{22}\,y_{22}\,z_{22}.
\]
The tensor rank of $\langle 2,2,2 \rangle_\pert$ is equal to the smallest number $r$ such that $f_\pert(X,Y,Z)$ can be written as a sum $\sum_{\rho=1}^r u_\rho(X) v_\rho(Y) w_\rho(Z)$, where $u_\rho$ is a linear form in the $x_{ij}$, similarly for $v_\rho$ and $w_\rho$.

\begin{proof}[\bfseries\upshape Proof of \cref{subth}; substitution method]
%Let $X = (x_{ij})_{ij\in[2]}$, $Y = (y_{ij})_{ij\in[2]}$ and $Z = (z_{ij})_{ij\in[2]}$ be matrices of variables. View $\langle 2,2,2\rangle_\alpha$ as a perturbed version of the trace of $XYZ$,
%\begin{multline*}
%f(X,Y,Z)_\alpha \coloneqq \sum_{j\in [2]} \bigl( x_{1j} y_{j1} z_{11} + x_{2j} y_{j1} z_{12} + x_{1j} y_{j2} z_{21}\bigr)\\ + x_{21} y_{12} z_{22} + \alpha\,x_{22} y_{22} z_{22}
%\end{multline*}
%and 
Suppose that the function $f_\pert(X,Y,Z)$ has rank~$r$ in the sense that it has a decomposition into a sum of~$r$ products of three linear forms as described above.
If $U = (u_{ij})_{ij\in[2]}$ is any upper triangular matrix, then $f_\pert(X,YU^{-1},UZ)$ has rank at most  $r$ and by direct computation
\[
f_\pert(X,YU^{-1},UZ) = f_\pert(X,Y,Z) + \frac{u_{12}}{u_{11}} (\pert-1)\, x_{22}\, y_{21}\, z_{22}.
\]
There exists an upper triangular matrix~$U$ such that the function $g_\pert(X,Y,Z) \coloneqq f_\pert(X,YU^{-1}, UZ)$ has a decomposition
\begin{equation}\label{eq:decomp}
g_\pert(X,Y,Z) = \sum_{\rho=1}^r u_\rho(X) v_\rho(Y) w_\rho(Z)
\end{equation}
in which $w_r(Z)$ is of the form $z_{21} + a_{12} z_{12} + a_{22} z_{22}$ for some $a_{12},a_{22} \in \FF$. %has $\begin{psmallmatrix}
%0 & *\\
%1 & *
%\end{psmallmatrix}$
%as a coefficient matrix with respect to the basis elements $z_{ij}$, where $*$ denotes any element of $\FF$. 
%Suppose $w_r(Z) = z_{21} - \tilde{w}(Z)$. 
Apply the substitution $z_{21} \mapsto \tilde w(Z) \coloneqq -  a_{12} z_{12} - a_{22} z_{22}$ to \eqref{eq:decomp} to see that %we  obtain %the following function which has rank at most $r-1$ %kill a summand in the decomposition~\eqref{eq:decomp}. We still compute
\begin{multline}\label{eq:decomp1}
\sum_{j\in [2]} \bigl(x_{1j}\, y_{j1}\, z_{11} + x_{2j}\, y_{j1}\, z_{12} + x_{1j}\, y_{j2}\, \tilde{w}(Z)  \bigr)\\ + x_{21} \,y_{12}\, z_{22} + \pert\, x_{22} \,y_{22}\, z_{22} + \frac{u_{12}}{u_{11}} (\pert-1)\, x_{22}\, y_{21}\, z_{22}\\
= \sum_{\rho=1}^{r-1} u_\rho(X) v_\rho(Y) w_\rho\Bigl(\begin{bsmallmatrix}z_{11} & z_{12}\\ \tilde{w}(Z) & z_{22}\end{bsmallmatrix} \Bigr).
\end{multline}
We can test that $y_{22}$ occurs in the obtained decomposition of \eqref{eq:decomp1} by setting $x_{22}, z_{22}$ to $1$ and $y_{21}$ to $0$ and the other $x_{ij}, z_{ij}$ to $0$. 
We can test that $y_{12}$ occurs in the obtained decomposition of \eqref{eq:decomp1} by setting $x_{21}, z_{22}$ to $1$ and the other $x_{ij}, z_{ij}$ to $0$. Say $y_{22}$ occurs in $v_{r-1}$ and $y_{12}$ occurs in $v_{r-2}$. Then, there is a substitution $y_{12} \mapsto \tilde{v}_{12}(Y)$, $y_{22}\mapsto \tilde{v}_{22}(Y)$,  which, applied to \eqref{eq:decomp1} yields
\begin{multline}\label{eq:decomp2}
\sum_{j\in [2]} \bigl( x_{1j}\, y_{j1}\, z_{11} + x_{2j}\, y_{j1}\, z_{12} + x_{1j} \,\tilde{v}_{j2}(Y)\, \tilde{w}(Z)\bigr)\\ + x_{21}\, \tilde{v}_{12}(Y)\, z_{22} + \pert\, x_{22}\, \tilde{v}_{22}(Y)\, z_{22} + \frac{u_{12}}{u_{11}} (\pert-1)\, x_{12}\, y_{21}\, z_{22} \\
= \sum_{\rho=1}^{r-3} u_\rho(X) v_\rho\Bigl(\begin{bsmallmatrix}y_{11} & \tilde{v}_{12}(Y)\\ y_{21}& \tilde{v}_{22}(Y)\end{bsmallmatrix}\Bigr) w_\rho\Bigl(\begin{bsmallmatrix}z_{11}&z_{12}\\ \tilde{w}(Z) & z_{22}\end{bsmallmatrix}\Bigr).
\end{multline}
To clean up, setting $z_{22} \mapsto 0$ in \eqref{eq:decomp2} shows that
\begin{multline}\label{eq:decomp3}
\sum_{j\in [2]} x_{1j}\, y_{j1}\, z_{11} + x_{2j}\, y_{j1}\, z_{12} + x_{1j}\, \tilde{v}_{j2}(Y)\, \tilde{w}(Z)\\
= \sum_{\rho=1}^{r-3} u_\rho(X) v_\rho\Bigl(\begin{bsmallmatrix}y_{11} & \tilde{v}_{12}(Y)\\ y_{21}& \tilde{v}_{22}(Y)\end{bsmallmatrix}\Bigr) w_\rho\Bigl(\begin{bsmallmatrix}z_{11}&z_{12}\\ \tilde{w}(Z) & 0\end{bsmallmatrix}\Bigr).
\end{multline}
We can test that $x_{21}$  occurs in the obtained decomposition \eqref{eq:decomp3} by setting $y_{11}, z_{12}$ to  $1$ and the other $x_{ij}, z_{ij}$ to $0$. Similarly, we can test that $x_{22}$ occurs in the obtained decomposition \eqref{eq:decomp3} by setting $y_{21}, z_{12}$ to $1$ and the other $x_{ij}, z_{ij}$ to $0$. Say $x_{21}$ occurs in $u_{r-3}$ and $x_{22}$ occurs in $u_{r-4}$. We apply a substitution $x_{21} \mapsto \tilde{u}_{21}(X)$, $x_{22} \mapsto \tilde{u}_{22}(X)$ to see that
\begin{multline}\label{eq:decomp4}
\sum_{\smash{j\in [2]}} x_{1j}\, y_{j1}\, z_{11} + \tilde{u}_{2j}(X)\, y_{j1}\, z_{12} + x_{1j}\, \tilde{v}_{j2}(Y)\, \tilde{w}(Z)\\
= \sum_{\rho=1}^{r-5} u_\rho\Bigl(\begin{bsmallmatrix}x_{11}& x_{12} \\ \tilde{u}_{21}(X)& \tilde{u}_{22}(X)\end{bsmallmatrix}\Bigr) v_\rho\Bigl(\begin{bsmallmatrix}y_{11} & \tilde{v}_{12}(Y)\\ y_{21}& \tilde{v}_{22}(Y)\end{bsmallmatrix}\Bigr) w_\rho\Bigl(\begin{bsmallmatrix}z_{11} & z_{12}\\ \tilde{w}(Z)&0\end{bsmallmatrix}\Bigr).
\end{multline}
Apply the substitution $z_{12} \mapsto 0$ to \eqref{eq:decomp4} to get
\begin{multline}
\sum_{\smash{j\in [2]}} x_{1j}\, y_{j1}\, z_{11}\\
= \sum_{\rho=1}^{r-5} u_\rho\Bigl(\begin{bsmallmatrix}x_{11} & x_{12}\\ \tilde{u}_{21}(X) & \tilde{u}_{22}(X)\end{bsmallmatrix}\Bigr) v_\rho\Bigl(\begin{bsmallmatrix} y_{11} & \tilde{v}_{12}(Y)\\ y_{21}& \tilde{v}_{22}(Y)\end{bsmallmatrix}\Bigr) w_\rho\Bigl(\begin{bsmallmatrix}z_{11} & 0 \\ \tilde{w}(Z) & 0\end{bsmallmatrix}\Bigr).
\end{multline}
which clearly has rank $2$. Therefore, $r\geq 7$. By \cref{paramreduction} we are done.
\end{proof}

\section{Border support rank}\label{sec:bsupprank}
%We denote the border rank of $t$ by $\borderrank(t)$ and the border support rank by $\bsupprank(t)$.
In this section all vector spaces are complex vector spaces. We will review a method that was introduced in \cite{landsberg2004ideals} to study equations for border rank and that was later used in \cite{hauenstein2013equations} to give a proof that $\borderrank(\langle 2,2,2\rangle) \geq 7$. Then, we will use this method to show that the border support rank of $\langle 2,2,2\rangle$ equals seven. Our Python code is included as an ancillary file with the arXiv submission.

View the space $\otimes^3 \CC^n$ as an affine variety, and let $\CC[\otimes^3\CC^n]$ be its coordinate ring. Define $\sigma_r\subseteq \otimes^3\CC^n$ as the subset of tensors with border rank at most $r$,
\[
\sigma_r \coloneqq \{s \in \otimes^3 \CC^n : \borderrank(s) \leq r\}.
\]
This is called the $r$th secant variety of the Segre variety of $\CC^n\times \CC^n \times \CC^n$. The set $\sigma_r$ is  Zariski closed in $\otimes^3 \CC^n$ by definition of border rank.  In other words, if we let $I(\sigma_r) \subseteq \CC[\otimes^3 \CC^n]$ be the ideal of polynomials on $\otimes^3 \CC^n$ that vanish identically on~$\sigma_r$, then $Z(I(\sigma_r)) = \sigma_r$. 

\subsection{Lower bounds by polynomials}
By definition, if $\borderrank(t) > r$ then there exists a polynomial in $I(\sigma_r)$ that does not vanish on~$t$. The following standard proposition says that we may in fact assume that this polynomial is homogeneous.

\begin{proposition}\label{sep}
Let $t\in \otimes^3 \CC^n$. If $\borderrank(t) > r$, then there exists a homogeneous polynomial~$f$ in~$I(\sigma_r)$ such that $f(t)\neq 0$.
\end{proposition}
\begin{proof} We give a proof for the convenience of the reader.
If $f(t)=0$ for all $f\in I(\sigma_r)$, then $t\in Z(I(\sigma_r)) = \sigma_r$, which is a contradiction. Let $f$ be a polynomial in $I(\sigma_r)$ such that $f(t)\neq0$. Let $f = \sum_d f_d$ be the decomposition of $f$ into homogeneous parts. There is a $d$ such that $f_d(t)\neq 0$.

Let $v\in \sigma_r$. For any $\alpha \in \CC$, define $g(\alpha) \coloneqq f(\alpha v)$. This is a polynomial in $\alpha$. We have $g(\alpha) = \sum_d \alpha^d f_d(v)$. Since $\sigma_r$ is closed under scaling and $f(v) = 0$, we have $g(\alpha) = 0$ for any $\alpha \in \CC$, so $g$ is the zero polynomial. Therefore, each coefficient $f_d(v)$ is 0. This argument holds for any $v\in \sigma_r$, so $f_d\in I(\sigma_r)$ for each $d$.
\end{proof}

The polynomial ring $\CC[\otimes^3\CC^n]$ decomposes into a direct sum of homogeneous parts $\CC[\otimes^3 \CC^n]_d$ and, by the above argument, the vanishing ideal $I(\sigma_r)$ decomposes accordingly as $I(\sigma_r) = \oplus_d I(\sigma_r)_d$ with~$I(\sigma_r)_d \subseteq \CC[\otimes^3 \CC^n]_d$.

%\subsection{Lower bounds by highest-weight vectors}

The space $\otimes^3 \CC^n$ has a natural action of $G\coloneqq\GL_n^{\times 3}$ and $\sigma_r$ is a $G$-submodule. Thus $\CC[\otimes^3 \CC^n]_d \cong \Sym^d (\otimes^3 (\CC^n)^*)$ has a natural action of $G$ and $I(\sigma_r)_d$ is a $G$-submodule. We will use the well-known theory of highest-weight vectors to exploit this symmetry. The theory of highest-weight vectors holds in a much more general setting than we need here. We refer to \cite[III.1.5]{kraft1984geometrische} and \cite{hall2015lie} for the general theory, and focus on a description of the theory for the group $\GL_n^{\times 3}$.
%As a $G$-module, $\CC[\otimes^3 \CC^n]$ decomposes according to degree: $\CC[\otimes^3 \CC^n] = \oplus_d\, \CC[\otimes^3 \CC^n]_d$. Each homogeneous part $\CC[\otimes^3 \CC^n]_d$ is isomorphic to $\Sym^d \otimes^3 (\CC^n)^*$ as a $G$-module. Also the ideal $I$ decomposes according to degree: $I = \oplus_d I_d$ with $I_d \subseteq \CC[\otimes^3 \CC^n]_d$, which directly follows from the proof of \cref{sep}.

%\begin{definition}
Let $W$ be a finite-dimensional $G$-module. Choose a basis so that $G$ becomes the group of triples of invertible matrices. Let $T\subseteq G$ be the subgroup of triples of diagonal matrices. For 
\[
t = (\diag(a_1,\ldots,a_n),\,\diag(b_1,\ldots,b_n),\,\diag(c_1,\ldots,c_n)) \in T
\]
and $z = (u,v,w)\in (\ZZ^n)^3$ define $t^z \coloneqq \prod_{i=1}^n a_i^{u_i} b_i^{v_i} c_i^{w_i}$. As a $T$-module, $W$ decomposes into weight spaces,
\[
W = \!\!\bigoplus_{z\in(\ZZ^n)^3} \!\! W_z \qquad \textnormal{where}\quad W_z = \{w\in W : t\cdot w = t^z w \,\,\forall\,t\in T\}.
\]
The vectors in $W_z$ are said to have \emph{weight} $z$. Let $U\subseteq G$ be the subgroup of triples of unipotent matrices, that is, upper triangular matrices with ones on the diagonal. A nonzero vector $v \in W_z$ is a \emph{highest-weight vector} if $u\cdot v = v$ for all $u\in U$.

%The following proposition summarises the important role that highest-weight vectors play in the representation theory of $\GL_n^{\times 3}$.

%\begin{proposition}\label{hwv}
A finite-dimensional (rational) representation $W$ of $\GL_n^{\times 3}$ is irreducible if and only if it has a unique highest-weight vector $v$, up to multiplication by a scalar, that is, $[W]^U = \Span_\CC v$. If $W$ is irreducible and $v$ is a highest-weight vector, then one has $W = \Span_\CC(G v)$. Moreover, two irreducible representations are isomorphic if and only if their highest-weight vectors have the same weight. %It turns out that the set of possible weights is precisely the set of triples of partitions with at most $n$ parts.
We call a sequence of $n$ nonincreasing integers a generalized partition.  It turns out that the weight of a highest-weight vector is a triple of generalized partitions. For any triple of generalized partitions $\lambda$, we will denote an abstract realisation of the $G$-module with highest-weight $\lambda$ by $V_{\lambda}$.
For any finite-dimensional $G$-module $W$, the highest-weight vectors in $W$ of weight $\lambda$ form a vector space, which we denote by~$[W_\lambda]^U$.

For a generalized partition $\lambda$, define the dual partition $\lambda^*$ as the generalized partition obtained from $\lambda$ by negating every entry and reversing the order. Then $V_{\lambda^*} = V_\lambda^*$, the dual module. We note that the polynomial irreducible representations are precisely the ones that are isomorphic to $V_\lambda$ with $\lambda$ a partition.

Recall that $\sigma_r$ is the variety of tensor in $\otimes^3 \CC^n$ of border rank at most~$r$. Consider the isotypic decomposition of $W\coloneqq \Sym^d (\otimes^3 (\CC^n)^*)$ and $I(\sigma_r)_d$ under the action of $\GL_n^{\times 3}$,
\begin{alignat*}{2}
W \,&=\, \bigoplus_{\lambda \vdash d}\, W_{\lambda^*} \,\,&=& \,\,\bigoplus_{\lambda \vdash d}\,\,k(\lambda)\,\, V_{\lambda}^*,\\
I(\sigma_r)_d \,&=\, \bigoplus_{\lambda\vdash d} I(\sigma_r)_{\lambda^*} \,\,&=& \,\,\bigoplus_{\lambda \vdash d}\,\,m(\lambda)\,\, V_{\lambda}^*,
\end{alignat*}
where $\lambda$ runs over all \emph{triples} of partitions of $d$ with at most $n$ parts, and $k(\lambda)\, V_{\lambda}^*$ denotes an isotypic component consisting of a direct sum of $k(\lambda)$ copies of the irreducible $G$-representation $V_{\lambda}^*$, similarly for $m(\lambda)\, V_{\lambda}^*$. Note that, although the direct sums run over triples of partitions $\lambda$, the representations $W$ and $I(\sigma_r)$ are not polynomial since we take duals. The number $k(\lambda)$ is exactly the dimension of the highest-weight vector space $[W_{\lambda^*}]^U$, and the number $m(\lambda)$ is the dimension of the highest-weight vector space $[I(\sigma_r)_{\lambda^*}]^U$. The following proposition extends \cref{sep} by saying that we may assume that the polynomial we are looking for is a highest-weight vector, if we replace $t$ by a random point in its $G$-orbit.

\begin{proposition}\label{guarantee}
Let $t\in \otimes^3\CC^n$. If $\borderrank(t) > r$, then there exists a highest-weight vector $f\in I(\sigma_r)$ and a group element $g\in G$ such that $f(gt) \neq 0$. 
\end{proposition}
\begin{proof}
We provide the proof for the convenience of the reader.
By \cref{sep}, there exists a homogeneous polynomial $f\in I(\sigma_r)$ such that we have $f(t)\neq 0$. By highest-weight theory, the polynomial $f$ can be written as a sum $\sum_{\lambda, i} g_{\lambda, i} f_{\lambda, i}$, where $f_{\lambda, i}$ is a highest-weight vector of type $\lambda$ in $I(\sigma_r)$ and $g_{\lambda,i}\in G$. Since $f(t)\neq 0$, there exists a $\lambda$ and an~$i$ so that $f_{\lambda,i}(g_{\lambda,i}^{-1}\, t) \neq 0$.
\end{proof}

\subsection{Highest-weight vector method} The following method was first proposed in \cite{landsberg2004ideals} to study equations for border rank and was later used in \cite{hauenstein2013equations} to give a proof that $\borderrank(\langle 2,2,2\rangle) \geq 7$. Let $t\in \otimes^3 \CC^n$ be a tensor for which we want to show $\borderrank(t) > r$. 
\begin{enumerate}
\item Choose a degree $d \in \NN$. Let $W$ be the space $\Sym^d (\otimes^3 (\CC^n)^*)$. Choose a partition triple $\lambda \vdash d$ such that the highest-weight vector space $[W_{\lambda^*}]^U$ is nonzero.
\item Construct a basis $b_1, \ldots, b_k$ for $[W_{\lambda^*}]^U$.
\item Find a linear combination $f$ of the basis elements $b_1, \ldots, b_k$ that vanishes on all tensors of border rank at most $r$, that is, $f \in [I(\sigma_r)_{\lambda^{*}}]^U$ where $\sigma_r$ is the variety of tensors with border rank at most $r$.
\item Show that $f$ does not vanish on $gt$ for some $g\in G$.
\end{enumerate}
The above method is guaranteed to work by \cref{guarantee}. Before applying the method, we will consider each step in more detail. 

\textbf{Step 1. Kronecker coefficient.} The dimension of the space of $U$-invariants $[(\Sym^d( \otimes^3(\CC^n)^*))_{\lambda^*}]^U$ is the so-called \emph{Kronecker coefficient} $k(\lambda)$. We pick a partition triple $\lambda$ such that the number $k\coloneqq k(\lambda)$ is nonzero. Algorithms for computing Kronecker coefficients have been implemented in for example Schur~\cite{schur}, Sage~\cite{sage} and the Python package Kronecker~\cite{kronecker}.

\textbf{Step 2. Las Vegas construction of basis.} For any natural number $\ell\leq n$, let $\phi_\ell \coloneqq e_{1}^* \wedge \cdots \wedge e_{\ell}^*$ be the Slater determinant living in $\wedge^\ell (\CC^n)^*$. For any partition $\mu \vdash d$ with at most $n$ parts, we let $\phi_\mu$ denote the tensor $\phi_{\nu_1} \otimes \cdots \otimes \phi_{\nu_{\mu_1}}$ living in $\otimes^d (\CC^n)^*$,
where $\nu$ denotes the transpose of $\mu$. Let $\lambda=(\lambda^{(1)},\lambda^{(2)}, \lambda^{(3)})$ be a triple of partitions of $d$. We define $\phi_{\lambda} \coloneqq \phi_{\lambda^{(1)}}\otimes\phi_{\lambda^{(2)}}\otimes\phi_{\lambda^{(3)}}$. This tensor lives in $\otimes^3 \otimes^d (\CC^n)^*$, but we view it as a tensor in $\otimes^d \otimes^3 (\CC^n)^*$ via the canonical reordering. Let $P_d$ be the canonical symmetrizer $\otimes^d \otimes^3 (\CC^n)^* \to \Sym^d (\otimes^3 (\CC^n)^*)$ acting from the right. The group $S_d^{\times 3}$ has a natural right action on $\otimes^3 \otimes^d (\CC^n)^*$ and via the reordering also on $\otimes^d \otimes^3 (\CC^n)^*$.
%
%\begin{proposition}[{\cite[4.2.17]{ikenmeyer2013geometric}}]
Let $\lambda$ be a triple of partitions of $d$.
The tensors $\{\phi_{\lambda} \pi P_d : \pi \in S_d^{\times 3}\}$ span the vector space $[(\Sym^d (\otimes^3(\CC^n)^*))_{\lambda^*}]^U$, see \cite[4.2.17]{ikenmeyer2013geometric}.
%\end{proposition}

We construct a basis of $[(\Sym^d (\otimes^3(\CC^n)^*))_{\lambda^*}]^U$ as follows. Randomly pick $k$ permutation pairs $\tau_1, \ldots, \tau_k \in S_d^{\times 2}$. Let $e\in S_d$ be the identity permutation. Let $\pi_i = (e, \tau_i^{(1)}, \tau_i^{(2)})$ and let $b_i \coloneqq \phi_\lambda \pi_i P_d$. Pick~$k$ random tensors $w_1,\ldots, w_k$ in $\otimes^3\CC^n$ and evaluate every $b_i$ in every $w_j$, giving a $k$-by-$k$ evaluation matrix~$M$. If $M$ has full rank, then $(b_1,\ldots, b_k)$ is the desired basis.

Before going to the next step we discuss how to efficiently implement the evaluation of a polynomial represented by a pair of permutations, as was already described in \cite{hauenstein2013equations}. Let $f = \phi_\lambda \pi P_d$ and let $t$ be the tensor $\sum_{i=1}^r t^1_i\otimes t^2_i \otimes t^3_i$ in $\otimes^3 \CC^n$. The evaluation of the polynomial $f$ at $t$ is equal to the contraction
\begin{align*}
\phi_\lambda \pi P_d\, t^{\otimes d} &= \phi_\lambda \pi \, t^{\otimes d}\\ &= \sum_{j \in [r]^d} \phi_\lambda \pi\, (t^1_{j_1}\otimes t^2_{j_1} \otimes t^3_{j_1}) \otimes \cdots \otimes  (t^1_{j_d}\otimes t^2_{j_d} \otimes t^3_{j_d})\\
 &= \sum_{j \in [r]^d} \phi_{\lambda^{(1)}} (t^1_{j_1} \otimes \cdots\otimes t^1_{j_d}) \\[-1em]
&\hspace{4.5em}\cdot \phi_{\lambda^{(2)}}\, \tau^{(1)} (t^2_{j_1} \otimes \cdots\otimes t^2_{j_d}) \\
&\hspace{4.5em}\cdot\phi_{\lambda^{(3)}}\, \tau^{(2)} (t^3_{j_1} \otimes \cdots\otimes t^3_{j_d}).
\end{align*}
Note that the last expression is a sum of a product of determinants. Let us study the first factor of a summand. Let $\nu$ denote the transpose of $\lambda^{(1)}$. We have
\begin{align*}
\phi_{\lambda^{(1)}} (t^1_{j_1} \otimes \cdots\otimes t^1_{j_d}) &=
(\phi_{\nu_1} \otimes \cdots \otimes \phi_{\nu_{\mu_1}}) (t^1_{j_1} \otimes \cdots\otimes t^1_{j_d})\\
&=\dett_{\nu_1}( t^1_{j_1}, \ldots, t^1_{j_{\nu_{1}}})
\dett_{\nu_2}( t^1_{j_{\nu_1}}, \ldots, t^1_{j_{\nu_{1}+\nu_{2}}})\\
&\hspace{2em}\cdots{} \dett_{\nu_{\mu_1}}( t^1_{j_{d-\nu_{\mu_1}}}, \ldots, t^1_{j_{d}}),
\end{align*}
where $\det_m(v_1,\ldots,v_m)$ denotes top $m$-by-$m$ minor of the matrix with columns $v_1, \ldots, v_m$.
Suppose that, in our evaluation of $\sum_{j}$, we have chosen values for $j_1,\ldots, j_{\nu_1}$ and suppose $\dett_{\nu_1}( t^1_{j_1}, \ldots, t^1_{j_{\nu_{1}}})$ is 0. Then whatever choices we make for $j_{\nu_1+1},\ldots, j_d$, the summand at hand will be zero. Recognizing this situation early is crucial.

\textbf{Step 3. Construction of a vector in $I(\sigma_r)$.}
Pick $k$ random tensors $t_1, \ldots, t_k$ of rank~$r$. Evaluate each basis element $b_i$ in each random tensor $t_j$. If the resulting matrix $(b_i(t_j))_{i,j\in [k]}$ has a nontrivial kernel, then we find a candidate highest-weight vector $f$ in~$I(\sigma_r)$. We can verify the correctness of the candidate by evaluating $f$ at a symbolic tensor of rank $r$. This evaluation should be zero. The way we do this symbolic evaluation is by working in $\otimes^3 \CC^6$ and using the straightening algorithm, see e.g.~the SchurFunctors package in Macaulay2~\cite{schurfunctors}. We used multi-prolongation to split up the computation in order to save memory. We refer to \cite{landsberg2004ideals, raicu2012secant} for a discussion of multi-prolongation.

\textbf{Step 4. Evaluating at $gt$.} Evaluate $f$ at $gt$ for a random $g\in G$. (In our case, it turns out that taking $g$ to be the identity is good enough.)

\subsection{The matrix multiplication tensor}
We will now prove that $\bsupprank(\langle 2,2,2 \rangle) = 7$.

\begin{proof}[\bfseries\upshape Proof of \cref{mainth}] The upper bound follows from \cref{subth}, so it remains to prove the lower bound. Let $\sigma_6$ be the variety of tensors in $\otimes^3 \CC^4$ of border rank at most~6. We will apply the method described above to the tensor $\langle 2,2,2\rangle_\pert$, see \cref{perturbed}.

%\textbf{Partition $(5,5,5,5)^3 \vdash 20$.}
 Let $d = 20$ and let $\lambda$ be the partition triple $(5,5,5,5)^{3}$. The Kronecker coefficient $k(\lambda)$ equals 4. Let $W\coloneqq \Sym^{20} (\otimes^3 (\CC^4)^*)$ and denote by $W_{\lambda^*}$ the isotypic component of type~$\lambda^*$. Writing permutations in the one-line notation, the following pairs of permutations define a basis $(b_1,b_2,b_3,b_4)$ for the highest-weight vector space $[W_{\lambda^*}]^U$:
\begin{align*}
\pi_1 = (&[5, 14, 8, 2, 12, 0, 1, 15, 6, 11, 18, 13, 4, 3, 9, 17, 7, 10, 16, 19],\\ &[14, 5, 9, 0, 6, 13, 16, 15, 4, 11, 3, 10, 12, 8, 2, 17, 7, 19, 18, 1]),\\[0.5em]
\pi_2 = (&[11, 18, 2, 12, 10, 5, 1, 17, 19, 9, 3, 4, 7, 6, 13, 0, 14, 16, 15, 8],\\ &[19, 1, 2, 7, 8, 3, 13, 6, 17, 10, 18, 12, 15, 4, 5, 11, 16, 0, 14, 9]),\\[0.5em]
\pi_3 = (&[2, 16, 17, 1, 4, 0, 7, 5, 10, 14, 11, 6, 18, 15, 9, 12, 19, 13, 3, 8],\\ &[15, 9, 0, 11, 19, 16, 18, 7, 2, 13, 5, 6, 17, 14, 8, 1, 12, 4, 10, 3]),\\[0.5em] 
\pi_4 = (&[9, 12, 14, 2, 6, 19, 18, 3, 15, 0, 1, 5, 11, 17, 7, 16, 8, 4, 13, 10],\\ &[14, 4, 18, 3, 11, 16, 15, 12, 5, 0, 17, 2, 10, 9, 13, 19, 7, 6, 1, 8]).
\end{align*}
%\pi_1 = (&[2, 4, 10, 12, 0, 7, 15, 19, 1, 17, 13, 18, 11, 8, 6, 16, 3, 5, 14, 9],\\ &[13, 16, 15, 4, 18, 3, 11, 1, 8, 7, 12, 6, 10, 19, 0, 9, 2, 17, 14, 5]),\\[0.5em]
%\pi_2 = (&[9, 18, 3, 11, 7, 0, 5, 16, 4, 17, 13, 14, 19, 2, 15, 12, 6, 10, 8, 1],\\ &[10, 11, 15, 16, 17, 18, 6, 12, 2, 14, 4, 8, 7, 9, 3, 13, 1, 0, 19, 5]),\\[0.5em]
%\pi_3 = (&[1, 9, 15, 4, 14, 6, 8, 17, 2, 0, 5, 10, 18, 3, 7, 11, 16, 13, 12, 19],\\ &[13, 9, 18, 0, 11, 4, 6, 10, 17, 8, 3, 1, 14, 5, 15, 16, 2, 19, 7, 12]),\\[0.5em]
%\pi_4 = (&[2, 14, 5, 9, 16, 1, 6, 3, 8, 4, 7, 19, 11, 12, 15, 10, 17, 0, 18, 13],\\ &[14, 2, 19, 1, 12, 6, 9, 18, 15, 0, 3, 10, 5, 17, 4, 11, 7, 8, 13, 16]).
%\end{align*}
The polynomial $f_{20} = 11832 g_1 + 233074 g_2 + 34117 g_3 - 32732 g_4$ is the only linear combination of the basis elements that is in $I(\sigma_6)$, up to scaling. We verified that $f_{20}$ is indeed in $I(\sigma_6)$ with the straightening algorithm. Evaluating $f_{20}$ on $\langle 2,2,2 \rangle_\pert$ yields
\[
f_{20}(\langle 2,2,2 \rangle_\pert) = -730140480 (\pert + 1) \pert^2.
\]

Let $d = 19$ and let $\lambda$ be the partition triple $(5,5,5,4)^3$. The Kronecker coefficient $k(\lambda)$ equals 31. Let $W\coloneqq \Sym^{19} (\otimes^3 (\CC^4)^*)$ and denote by $W_{\lambda^*}$ the isotypic component of type $\lambda^*$. The following pairs of permutations define a basis $(b_1,\ldots,b_{31})$ for the highest-weight vector space $[W_{\lambda^*}]^U$:

{\small
\begin{align*}
\pi_1 = (&[4, 8, 13, 3, 1, 12, 5, 11, 9, 15, 2, 7, 0, 17, 14, 6, 10, 18, 16], \\ &[2, 18, 5, 7, 9, 13, 0, 12, 1, 15, 10, 8, 4, 11, 16, 3, 17, 6, 14]), \displaybreak[0]\\[0.5em]
\pi_2 = (&[12, 15, 11, 7, 2, 6, 8, 17, 9, 1, 16, 13, 4, 0, 3, 10, 18, 14, 5], \\ &[11, 9, 14, 0, 15, 13, 16, 3, 6, 8, 17, 7, 10, 5, 18, 2, 12, 1, 4]), \displaybreak[0]\\[0.5em]
\pi_3 = (&[14, 1, 2, 15, 6, 3, 7, 13, 4, 18, 8, 9, 12, 10, 16, 5, 17, 0, 11], \\ &[7, 18, 2, 10, 4, 12, 0, 9, 15, 6, 5, 13, 1, 17, 14, 16, 8, 3, 11]), \displaybreak[0]\\[0.5em]
\pi_4 = (&[4, 1, 0, 12, 7, 13, 9, 16, 6, 8, 18, 15, 17, 11, 14, 2, 10, 3, 5], \\ &[5, 13, 17, 14, 3, 4, 6, 11, 8, 18, 1, 15, 2, 0, 9, 16, 7, 10, 12]), \displaybreak[0]\\[0.5em]
\pi_5 = (&[11, 14, 5, 0, 15, 8, 2, 17, 1, 13, 4, 9, 16, 6, 7, 10, 18, 3, 12], \\ &[8, 18, 4, 14, 6, 16, 10, 2, 11, 9, 5, 0, 13, 12, 1, 7, 3, 17, 15]), \displaybreak[0]\\[0.5em]
\pi_6 = (&[10, 5, 18, 8, 15, 2, 16, 1, 0, 13, 3, 4, 7, 14, 11, 6, 12, 17, 9], \\ &[0, 8, 12, 2, 3, 9, 11, 13, 5, 1, 14, 7, 4, 16, 17, 18, 15, 10, 6]), \displaybreak[0]\\[0.5em]
\pi_7 = (&[12, 1, 11, 16, 13, 7, 2, 17, 10, 15, 3, 0, 5, 4, 14, 6, 9, 8, 18], \\ &[8, 1, 4, 2, 12, 14, 18, 15, 7, 9, 0, 11, 3, 10, 6, 17, 13, 5, 16]), \displaybreak[0]\\[0.5em]
\pi_8 = (&[17, 18, 6, 11, 4, 2, 1, 9, 15, 16, 5, 8, 10, 0, 12, 13, 3, 14, 7], \\ &[14, 1, 18, 6, 10, 15, 3, 5, 11, 16, 12, 9, 13, 7, 0, 17, 8, 4, 2]), \displaybreak[0]\\[0.5em]
\pi_9 = (&[8, 2, 10, 3, 6, 4, 11, 18, 13, 0, 5, 1, 15, 17, 12, 16, 14, 7, 9], \\ &[2, 5, 13, 16, 1, 10, 3, 14, 4, 17, 18, 12, 0, 11, 9, 6, 7, 8, 15]), \displaybreak[0]\\[0.5em]
\pi_{10} = (&[13, 17, 15, 1, 12, 0, 9, 10, 6, 18, 7, 16, 14, 5, 2, 4, 11, 8, 3], \\ &[6, 12, 11, 10, 2, 14, 13, 0, 9, 15, 16, 17, 5, 8, 3, 7, 1, 18, 4]), \displaybreak[0]\\[0.5em]
\pi_{11} = (&[14, 5, 4, 1, 16, 8, 3, 7, 10, 13, 18, 6, 2, 17, 11, 9, 15, 12, 0], \\ &[5, 9, 10, 1, 2, 4, 14, 18, 8, 11, 7, 6, 15, 17, 16, 3, 0, 13, 12]), \displaybreak[0]\\[0.5em]
\pi_{12} = (&[1, 5, 4, 13, 15, 2, 17, 16, 8, 10, 11, 6, 7, 3, 12, 14, 9, 0, 18], \\ &[9, 5, 7, 8, 6, 11, 18, 3, 10, 4, 14, 17, 13, 0, 12, 15, 16, 1, 2]), \displaybreak[0]\\[0.5em]
\pi_{13} = (&[16, 13, 4, 3, 5, 2, 1, 15, 18, 6, 12, 0, 14, 8, 17, 7, 10, 11, 9], \\ &[2, 7, 8, 18, 16, 4, 6, 14, 0, 15, 9, 5, 1, 12, 10, 13, 17, 11, 3]), \displaybreak[0]\\[0.5em]
\pi_{14} = (&[5, 12, 0, 9, 3, 7, 17, 2, 6, 14, 11, 8, 15, 4, 1, 10, 13, 18, 16], \\ &[5, 15, 18, 8, 17, 11, 9, 4, 13, 1, 16, 2, 0, 14, 7, 10, 12, 3, 6]), \displaybreak[0]\\[0.5em]
\pi_{15} = (&[12, 6, 9, 14, 18, 5, 17, 2, 1, 4, 3, 11, 0, 10, 15, 7, 16, 13, 8], \\ &[9, 1, 16, 18, 14, 5, 6, 0, 10, 13, 3, 7, 15, 4, 11, 17, 12, 2, 8]), \displaybreak[0]\\[0.5em]
\pi_{16} = (&[1, 18, 4, 8, 5, 3, 0, 16, 6, 10, 11, 2, 17, 7, 9, 12, 14, 13, 15], \\ &[8, 2, 15, 12, 18, 6, 0, 11, 13, 5, 9, 4, 16, 7, 10, 17, 14, 1, 3]), \displaybreak[0]\\[0.5em]
\pi_{17} = (&[18, 8, 16, 6, 5, 7, 2, 13, 0, 4, 12, 11, 14, 15, 3, 17, 1, 10, 9], \\ &[12, 9, 14, 2, 18, 5, 0, 13, 4, 16, 8, 7, 1, 10, 6, 3, 17, 11, 15]), \displaybreak[0]\\[0.5em]
\pi_{18} = (&[7, 5, 16, 15, 1, 0, 8, 11, 14, 17, 12, 6, 9, 3, 10, 18, 13, 4, 2], \\ &[8, 9, 0, 4, 2, 3, 5, 13, 18, 12, 6, 1, 16, 11, 17, 10, 14, 7, 15]), \displaybreak[0]\\[0.5em]
\pi_{19} = (&[2, 17, 0, 14, 15, 8, 1, 9, 12, 5, 10, 3, 7, 11, 4, 16, 6, 13, 18], \\ &[13, 3, 0, 15, 7, 17, 18, 10, 6, 16, 1, 8, 9, 14, 12, 4, 5, 2, 11]), \displaybreak[0]\\[0.5em]
\pi_{20} = (&[0, 16, 9, 3, 15, 1, 4, 14, 7, 2, 18, 10, 12, 11, 17, 8, 6, 5, 13], \\ &[3, 2, 13, 11, 8, 1, 5, 4, 0, 16, 7, 17, 6, 12, 14, 9, 18, 15, 10]), \displaybreak[0]\\[0.5em]
\pi_{21} = (&[17, 3, 5, 14, 0, 16, 2, 8, 1, 11, 7, 18, 12, 6, 9, 15, 4, 13, 10], \\ &[7, 2, 17, 8, 0, 13, 6, 1, 4, 5, 18, 9, 15, 10, 16, 11, 3, 14, 12]),\displaybreak[0] \\[0.5em]
\pi_{22} = (&[5, 4, 1, 14, 16, 3, 9, 17, 12, 8, 2, 6, 11, 7, 18, 15, 13, 0, 10], \\ &[6, 14, 8, 7, 9, 18, 3, 12, 15, 2, 0, 1, 13, 5, 10, 16, 4, 11, 17]), \displaybreak[0]\\[0.5em]
\pi_{23} = (&[17, 4, 10, 13, 14, 1, 6, 8, 5, 15, 9, 2, 0, 11, 18, 7, 3, 12, 16], \\ &[6, 3, 11, 12, 15, 17, 10, 2, 8, 5, 1, 0, 14, 7, 9, 18, 13, 4, 16]), \displaybreak[0]\\[0.5em]
\pi_{24} = (&[3, 9, 0, 15, 14, 7, 1, 16, 2, 8, 11, 4, 17, 12, 10, 6, 18, 13, 5], \\ &[10, 11, 3, 2, 1, 9, 14, 13, 18, 16, 0, 4, 15, 8, 5, 12, 6, 7, 17]), \displaybreak[0]\\[0.5em]
\pi_{25} = (&[12, 2, 8, 6, 16, 1, 15, 9, 11, 14, 10, 3, 5, 17, 0, 13, 18, 4, 7], \\ &[8, 2, 14, 1, 6, 17, 16, 3, 7, 9, 11, 12, 18, 0, 5, 13, 15, 10, 4]), \displaybreak[0]\\[0.5em]
\pi_{26} = (&[2, 16, 14, 6, 9, 0, 11, 12, 3, 15, 1, 18, 17, 7, 4, 8, 13, 5, 10], \\ &[10, 15, 13, 12, 17, 0, 16, 7, 4, 11, 1, 2, 6, 14, 8, 5, 9, 3, 18]), \displaybreak[0]\\[0.5em]
\pi_{27} = (&[10, 7, 6, 0, 12, 11, 16, 13, 1, 3, 17, 14, 8, 18, 4, 2, 9, 5, 15], \\ &[3, 17, 11, 12, 6, 5, 2, 13, 18, 14, 9, 1, 7, 16, 4, 8, 10, 15, 0]), \displaybreak[0]\\[0.5em]
\pi_{28} = (&[16, 6, 8, 4, 7, 5, 9, 1, 0, 2, 14, 13, 17, 10, 18, 15, 11, 3, 12], \\ &[6, 11, 1, 12, 2, 8, 5, 9, 3, 16, 15, 18, 4, 7, 14, 0, 10, 17, 13]), \displaybreak[0]\\[0.5em]
\pi_{29} = (&[8, 13, 7, 0, 17, 4, 2, 15, 16, 1, 18, 3, 5, 11, 12, 10, 6, 14, 9], \\ &[4, 13, 1, 10, 18, 12, 2, 5, 17, 7, 6, 15, 8, 9, 0, 11, 16, 14, 3]), \displaybreak[0]\\[0.5em]
\pi_{30} = (&[1, 6, 12, 0, 3, 10, 9, 13, 17, 4, 7, 8, 18, 14, 2, 5, 15, 16, 11], \\ &[16, 6, 10, 11, 15, 8, 17, 13, 14, 4, 5, 1, 3, 12, 2, 7, 0, 18, 9]), \displaybreak[0]\\[0.5em]
\pi_{31} = (&[5, 10, 11, 8, 17, 16, 2, 15, 12, 14, 0, 18, 3, 1, 7, 9, 6, 4, 13], \\ &[10, 15, 4, 12, 18, 3, 16, 6, 0, 13, 11, 7, 1, 8, 9, 2, 14, 17, 5]).
\end{align*}
}
Let 
\begin{alignat*}{2}
c_1 &= &289082199568614200505625810989998081122378290025627334\displaybreak[0]\\
c_2 &= &41448548699164679707399349100915823812613974963005402\displaybreak[0]\\
c_3 &= &211649838021887426162677078824519293749517217920047823\displaybreak[0]\\
c_4 &= &-118150576713220917823141541211872001702845422153137763\displaybreak[0]\\
c_5 &= &-71972591371289085208000082313759547126396087856917092\displaybreak[0]\\
c_6 &= &-148042611712972282129069557835544665097810271759437007\displaybreak[0]\\
c_7 &= &-20671385701071233448917086723379921457752823704368686\displaybreak[0]\\
c_8 &= &-41700697565765737458921317121977791710351222967960389\displaybreak[0]\\
c_9 &= &89818454969459149830510070194701368406615458716738371\displaybreak[0]\\
c_{10} &= &-33389561951163547125931836395846743479037338582546746\displaybreak[0]\\
c_{11} &= &-55953034618025281839233784369005651793756337420914611\displaybreak[0]\\
c_{12} &= &99436050816695444459576518293215696786461418941439932\displaybreak[0]\\
c_{13} &= &-30608800079918651823012662681016076665421200200986429\displaybreak[0]\\
c_{14} &= &62322369796163233078186315204176712499710334162812978\displaybreak[0]\\
c_{15} &= &71531123200873494604907676681446086219352685074695096\displaybreak[0]\\
c_{16} &= &11103950876950753893392891180499777390516447716768874\displaybreak[0]\\
c_{17} &= &-18170416924354926777786745151805158474424942420073625\displaybreak[0]\\
c_{18} &= &56636600557844043196391811853778001287738236566321291\displaybreak[0]\\
c_{19} &= &-49475697236538461568207568070821224602714314684182556\displaybreak[0]\\
c_{20} &= &-58897567946922439319826816178640661508235201647724834\displaybreak[0]\\
c_{21} &= &-29789369352552042959878217935401203848547004115080562\displaybreak[0]\\
c_{22} &= &42553086095082787553533988614363448520647296308373860\displaybreak[0]\\
c_{23} &= &-10584947869810207513601472123471095674362492708851758\displaybreak[0]\\
c_{24} &= &-155536179226293398590182659612811187764949236460651258\displaybreak[0]\\
c_{25} &= &-15163630056597008306009257387099740416829146255166469\displaybreak[0]\\
c_{26} &= &152468055855066906135282920200590542819196123610118125\displaybreak[0]\\
c_{27} &= &-170101205621738870358375711649013594303036219144235962\displaybreak[0]\\
c_{28} &= &-36619800006361115328892590783407206736313224654320560\displaybreak[0]\\
c_{29} &= &63636824324804825079032794300460871506246849887804488\displaybreak[0]\\
c_{30} &= &-114422655018015193150391631424350000645293977961135740\displaybreak[0]\\
c_{31} &= &99270978701207213884119395668714341424298017907910144
\end{alignat*}
and define $f_{19} = c_1 g_1 +  \cdots + c_{31}g_{31}$. This is the only linear combination that is in $I(\sigma_6)$, up to scaling. We verified that $f_{19}$ is in $I(\sigma_6)$ by straightening. Evaluating $f_{19}$ at $\langle 2,2,2 \rangle_\pert$ yields
\[
69332245782016022615247261570208505413020193878724712262 (3 \pert + 2) \pert.
\]
We have thus found two highest-weight vectors
\begin{align*}
%f_{19} &\in \HWV_{(5,5,5,4)^3*} I(V_6)_{19} \subseteq \Sym^{19} \otimes^3\! (\CC^4)^*\\
f_{19} &\in [I(\sigma_6)_{(5,5,5,4)^{3\,*}}]^U \subseteq \Sym^{19} (\otimes^3 (\CC^4)^*)\\
%f_{20} &\in \HWV_{(5,5,5,5)^3*} I(V_6)_{20} \subseteq \Sym^{20} \otimes^3\! (\CC^4)^*
f_{20} &\in [I(\sigma_6)_{(5,5,5,5)^{3\,*}}]^U \subseteq \Sym^{20} (\otimes^3 (\CC^4)^*)
\end{align*}
such that $f_{19}(\langle2,2,2\rangle_\pert) = \alpha\, \pert(3\pert+2)$ and $f_{20}(\langle2,2,2\rangle_\pert) = \beta\, \pert^2(\pert+1)$, where $\alpha$ and~$\beta$ are nonzero constants. The only simultaneous root of these polynomials occurs at $\pert=0$. This means that for any nonzero $\pert$, the point $\langle 2,2,2 \rangle_\pert$ is not contained in $\sigma_6$. From \cref{paramreduction} we conclude that the border rank of any tensor with the same support as $\langle 2,2,2 \rangle$ is at least seven, which proves the theorem.
\end{proof}

\begin{remark}The lower bound $\borderrank(\langle 2,2,2\rangle) \geq 7$ in \cite{hauenstein2013equations} was also obtained by showing that the highest-weight vector space $[I(\sigma_6)_{(5,5,5,5)^{3\,*}}]^U$ is nonzero, and the evaluation of a nonzero element $v\in [I(\sigma_6)_{(5,5,5,5)^{3\,*}}]^U$ at $\langle 2,2,2\rangle$ is nonzero.
\end{remark}

%From the above proof one readily obtains the following observation, which can be interpreted as saying that \cref{paramreduction} is optimal.
%
%\begin{corollary}
%The orbits $G \cdot \langle 2,2,2 \rangle_{-1}$, $G \cdot \langle 2,2,2 \rangle_{-2/3}$ and $G \cdot \langle 2,2,2 \rangle_z$ are pairwise disjoint for any $z\in \CC\setminus\{-1,-2/3\}$.
%\end{corollary}
%\begin{proof}
%Suppose $G \cdot \langle 2,2,2 \rangle_{-1}$ intersects $G\cdot \langle 2,2,2\rangle_z$ for a $z\in \CC\setminus\{0,-1\}$, say $g \cdot \langle 2,2,2 \rangle_z = \langle 2,2,2 \rangle_{-1}$. Then
%\[
%0 = f_{20}(\langle 2,2,2 \rangle_{-1}) = f_{20}(g\cdot \langle 2,2,2 \rangle_z) = (g^{-1}\cdot f_{20})(\langle 2,2,2 \rangle_z).
%\]
%However,  $I(V_6)_{(5,5,5,5)^3}$ is one-dimensional, so $g^{-1}\cdot f_{20}$ is a non-zero scalar multiple of~$f_{20}$ and therefore $(g^{-1}\cdot f_{20})(\langle 2,2,2 \rangle_z)$ is non-zero, which is a contradiction. The other cases are proved similarly.
%\end{proof}

\paragraph{Acknowledgements.} The authors are grateful to Christian Ikenmeyer for helpful discussions. 
MC acknowledges financial support from the European Research Council (ERC Grant Agreement no.~337603), the Danish Council for Independent Research (Sapere Aude), and VILLUM FONDEN via the QMATH Centre of Excellence (Grant no.~10059).
JZ is supported by NWO through the research programme 617.023.116. The computations in this work were carried out on the Dutch national e-infrastructure with the support of SURF Cooperative.

\raggedright
%\bibliographystyle{alphaabbr}
%\bibliography{all}

\bibliographystyle{alphaurlpp}
\bibliography{all}

\vspace{1em}
\textbf{Matthias Christandl}\\
QMATH, Department of Mathematical Sciences, University of Copenhagen, Universitetsparken 5, 2100 Copenhagen, Denmark.\\
Email: \href{mailto:christandl@math.ku.dk}{christandl@math.ku.dk}\\[1em]
\textbf{Markus Bläser}\\
Computer Science, Saarland University, Saarland Informatics Campus E1.3, 66123 Saarbrücken, Germany.\\
Email: \href{mailto:mblaeser@cs.uni-saarland.de}{mblaeser@cs.uni-saarland.de}\\[1em]
\textbf{Jeroen Zuiddam}\\
QuSoft, CWI Amsterdam and University of Amsterdam, Science Park~123, 1098~XG Amsterdam, Netherlands. \\
Email: \href{mailto:j.zuiddam@cwi.nl}{j.zuiddam@cwi.nl}

\end{document}